\documentclass[draftcls, onecolumn, 12pt]{IEEEtran}
\usepackage[T1]{fontenc}% optional T1 font encoding
\usepackage{graphicx}
\usepackage{amsthm}
\usepackage{amsmath}
\usepackage{epstopdf}
\usepackage[all]{xy}
\usepackage[cmintegrals]{newtxmath}

\newtheorem{proposition}{Proposition}

\newtheorem{definition}{Definition}
\newtheorem{remark}{Remark}
\newtheorem{theorem}{Theorem}
\newtheorem{example}{Example}
\newtheorem{lemma}{Lemma}

\ifCLASSINFOpdf
\else
\fi

\begin{document}

\title{Matched Metrics to the Binary Asymmetric Channels}

\author{Claudio~M.~Qureshi

\thanks{The author was supported by CNPq (grant 150270/2016-0) and FAPESP (grants 2015/26420-1 and 2013/25977-7) and is with the Institute of Mathematics, Statistics and Computing Science of the University of Campinas, SP , Brazil (email: cqureshi@ime.unicamp.br).}% <-this % stops a space
\thanks{Digital Object Identifier: 10.1109/TIT.2016.xxxxxxx}}

%Manuscript submitted to IEEE Transactions on Information Theory
\markboth{}%
{Shell \MakeLowercase{\textit{et al.}}: Bare Demo of IEEEtran.cls for IEEE Communications Society Journals}

\maketitle

\begin{abstract}
In this paper we establish some criteria to decide when a discrete memoryless channel admits a metric in such a way that the maximum likelihood decoding coincides with the nearest neighbor decoding. In particular we prove a conjecture presented by M. Firer and J. L. Walker establishing that every binary asymmetric channel admits a matched metric.
\end{abstract}


\IEEEpeerreviewmaketitle

\section{Introduction}  

As it is well known, maximum likelihood decoding (MLD) over a symmetric channel coincides with nearest neighbor decoding (NND) with respect to the Hamming metric. In this paper we deal with the problem of matching a metric to a given channel regarding the decoding criteria mentioned above. This problem was considered in 1967 by J. L. Massey \cite{Massey67} where a metric matched to a discrete memoryless channel is defined as a metric for which the NND is a MLD. Since then, this type of matching have been studied in some special cases. For instance, certain channels matching to the Lee metric were obtained in \cite{CW71}. The problem of matching a metric to a channel was taken up by G. S{\'e}guin \cite{Seguin80}, where the main focus was on sequences of additive metrics. In the referred paper it is used a stronger condition also assumed here: a metric matched to a channel is one for which not only the NND is a MLD but also the MLD is a NND. The author obtains necessary and sufficient conditions for the existence of additive metrics matched to a channel and raises the question of what happens if the restriction of additivity of the metric is removed. There was no significant progress until the paper \cite{FW16} by M. Firer and J. Walker, where the authors proved, among other results, the existence of a metric (not necessarily additive) matched to the Z-channels and to the $n$-fold binary asymmetric channel (BAC) for $n=2,3$, and conjectured that this is also true for $n>3$. Some recent progress in this direction was obtained in \cite{DF16} where it is presented an algorithm to decide if a channel is metrizable and in that case return a metric matched to the channel, and in \cite{Poplawski16} where the author proved that the BAC channels are metrizable in the weaker sense of J. L. Massey in \cite{Massey67}.

The main results of this paper are Theorem \ref{MainTheorem1}, which establishes a necessary and sufficient condition for metrizability of a channel in terms of graph theory, and Theorem \ref{MainTheorem2} which establishes that the BAC channels are metrizable. Other contributions are the association of channels with graphs (which allows the use of techniques from graph theory to approach problems related to channels) and the introduction of a new structure: the colored posets, which may also be useful in other contexts. This work is organized as follows: In Section \ref{SecPreliminaries} we give a brief review of definitions and concepts needed in the development of the paper. In Section \ref{SecGraphG1} we associate a graph with a channel and discuss some results of \cite{FW16} and \cite{Poplawski16} in terms of this graph. In Section \ref{SecColored} we introduce the concept of colored poset which is used to describe an algorithm for constructing a metric matched to a channel whenever it is metrizable. A necessary and sufficient criterion for metrizability of a channel is also derived. In Section \ref{SecBAC} this criterion is used to prove that the BAC channels are metrizable. In Section \ref{SecFinal} we introduce the concept of order of metrizability of a channel and settle some problems related to this.

\section{Preliminaries} \label{SecPreliminaries}

We summarize here some concepts and results to be used in the following sections.

A discrete memoryless channel (simply referred as channel in this paper) $W:\mathcal{X}\rightarrow \mathcal{X}$ is characterized by its transition matrix related to the input and output alphabet $\mathcal{X}=\{x_1,x_2,\ldots, x_{N}\}$. This matrix $[W] \in \mathcal{M}_{N\times N}(\mathbb{R})$ is given by $[W]_{ij}=\mbox{Pr}_{W}(x_i|x_j)$, the probability of receiving $x_{i}$ if $x_{j}$ was sent. When the channel is understood, this conditional probability is denoted by $\mbox{Pr}(x_i|x_j)$. Matrices associated with channels are characterized by the property that every entry is non-negative and the sum of the entries in each column is one.

A channel $W:\mathcal{X}\rightarrow \mathcal{X}$ is metrizable (in the strong sense of \cite{FW16} and \cite{Seguin80}) if there is a metric $d:\mathcal{X}\times \mathcal{X}\rightarrow [0,\infty)$ (i.e. $d$ is a definite-positive symmetric function satisfying the triangle inequality) such that every nearest neighbor decoder is a maximum likelihood decoder and vice versa. This is also equivalent to each of the following statements:
\begin{itemize}
\item[i)] For all $x\in \mathcal{X}$ and every code $C\subseteq \mathcal{X}$ we have $$ \mbox{arg }\max_{y\in C} \mbox{Pr}(x|y)=\mbox{arg }\min_{y\in C} d(x,y) ,$$ where both $\mbox{arg}\max$ and $\mbox{arg}\min$ are interpreted as returning lists of size at least $1$.
\item[ii)] For all $x,y,z \in \mathcal{X}$ the following condition holds: $$\mbox{Pr}(x|y)\leq \mbox{Pr}(x|z) \Leftrightarrow d(x,y)\geq d(x,z).$$
\end{itemize}

In this paper we only deal with reasonable channels (in the sense of \cite{FW16}), that is, channels $W:\mathcal{X}\rightarrow \mathcal{X}$ such that \begin{equation}\label{EqReasonable}
\mbox{Pr}(x|x)> \mbox{Pr}(x|y),\ \forall x, y \in \mathcal{X} \textrm{ with } y\neq x,
\end{equation}
which is a necessary condition for a channel to be metrizable.

Let $W:\mathcal{X} \rightarrow \mathcal{X}$ be a channel, $\mathcal{X}^{(2)}=\{A\subseteq \mathcal{X}: \#A=2\}$ be the family of $2$-subsets of $\mathcal{X}$ and $h:\mathcal{X}^{(2)}\rightarrow [0,+\infty)$ be a non-zero function. We say that $h$ is coherent-with-$W$ if $$\mbox{Pr}(x|y)\leq \mbox{Pr}(x|z) \Leftrightarrow  h(\{x,y\})\leq h(\{x,z\}),$$ for all $x,y,z \in \mathcal{X}$ with $x\neq y$ and $x\neq z$. If such a function exists, we can construct a metric matched to the channel as follows.

\begin{proposition}\label{PropMetricFromh}
Let $W:\mathcal{X}\rightarrow \mathcal{X}$ be a channel and $h:\mathcal{X}^{(2)}\rightarrow [0,+\infty)$ be a coherent-with-$W$ function with maximum value $m=\max\{h(x):x\in \mathcal{X}^{(2)}\}$. The function $d:\mathcal{X}\times \mathcal{X}\rightarrow [0,+\infty)$ given by:
\begin{equation}\label{EqMetricFromf}
d(x,y)=\left\{\begin{array}{ll}
2m-h\left(\{x,y\}\right) & \textrm{if }x\neq y, \\ 0 & \textrm{if }x=y,
\end{array} \right.
\end{equation}
is a metric matched to the channel $W$.
\end{proposition}

\begin{proof}
To prove that $d$ is positive-definite, we note that $d(x,y)=2m-h(\{x,y\})\geq m>0$. The function $d$ is clearly symmetric since $\{x,y\}=\{y,x\}$. To prove triangle inequality, we consider $x,y,z \in \mathcal{X}$ pairwise distinct and note that $d(x,y)+d(y,z)=4m-h(\{x,y\})-h(\{y,z\})\geq 2m\geq 2m-h(\{x,z\})=d(x,z)$. Therefore $d$ is a metric. This metric matches to the channel $W$ because it is reasonable and for $x,y,z \in \mathcal{X}$ pairwise distinct we have $\mbox{Pr}(x|y)\leq \mbox{Pr}(x|z) \Leftrightarrow  h(\{x,y\})\leq h(\{x,z\}) \Leftrightarrow d(x,y)\geq d(x,z)$, where the first equivalence is because $h$ is coherent-with-$W$ and the second by the definition of $d$.
\end{proof}

The binary ($1$-fold) asymmetric channel with parameters $(p,q)\in [0,1]^2$ (denoted by $BAC^{1}(p,q)$) is the channel with input and output alphabet $\mathbb{Z}_2=\{0,1\}$ and conditional probabilities $\mbox{Pr}_{1}(1|0)=p$ and $\mbox{Pr}_{1}(0|1)=q$ (and $\mbox{Pr}_{1}(0|0)=1-p$ and $\mbox{Pr}_{1}(1|1)=1-q$). The $n$-fold binary asymmetric channel $BAC^n(p,q)$ is the channel with input and output alphabet $\mathcal{X}=\mathbb{Z}_{2}^{n}$ and for $x=(x_1,\ldots, x_n)$ and $y=(y_1,\ldots,y_n)$ in $\mathbb{Z}_{2}^{n}$ the conditional probabilities are given by $$\mbox{Pr}(x|y)=\prod_{i=1}^{n}\mbox{Pr}_{1}(x_i|y_i).$$

We remark that the channel $BAC^{n}(p,q)$ verifies condition (\ref{EqReasonable}) if and only if $p+q<1$ (therefore only this case will be considered in this paper). Indeed, for $n=1$ it is obvious and for $n>1$ it is a direct consequence of Equation (\ref{EqCociente}) in Section \ref{SecBAC}. The metrizability of $BAC^n(p,q)$ was established in \cite{FW16} for the case $pq=0$ and $n$ arbitrary (the $n$-fold $Z$-channel) and for $p+q<1$ and $n=2,3$. The remaining case is when $p+q<1$ and $pq>0$. For this case, we prove that the corresponding channels are metrizable in Theorem \ref{MainTheorem2}.

A partially ordered set (or poset) is a pair $(P,\leq)$ where $\leq$ is a partial order relation (i.e. it is reflexive, antisymmetric and transitive). The poset is denoted by $P$ when the order relation is understood. Each poset is associated with a Hasse diagram, which is a representation of the poset in such a way that if $x<y$ the element $y$ is above $x$, and there is a segment connecting these points whenever there is no $z\in P$ with $x<z<y$. A directed graph (or digraph) $G$ is determined by a pair $(V,E)$ where $V$ is a set, called the vertex set, and $E\subseteq V\times V$ is the edge set. When $(v,w)\in E$ we say that the edge $v\to w$ belongs to $G$. A path in $G$ of length $r\geq 0$ is a finite sequence of vertices $c=(v_0,\ldots,v_r)$ such that $v_i\to v_{i+1}$ belongs to $G$ for $0\leq i <r$. If $c'=(w_0,\ldots,w_s)$ is other path in $G$ with $w_0=v_r$, the path $c*c':=(v_0,\ldots,v_r=w_0,w_1,\ldots,w_s)$ is also in $G$ and it is called the concatenation of $c$ and $c'$. The reverse path of $c$ is $\overline{c}=(v_r,\ldots,v_0)$, which is not necessarily a path in $G$. When $r\geq 1$ and $v_0=v_r$, the path $c=(v_0,\ldots,v_r)$ is called a directed cycle. A digraph without directed cycles is called acyclic. From a directed acyclic graph $G=(V,E)$ we have a natural poset structure on $V$ defining $x\leq y$ if there is a (directed) path (of length $r\geq 0$) from $x$ to $y$. When we refer to the Hasse diagram of a directed acyclic graph $G$ we mean the Hasse diagram of their associated poset.

\section{The graph $\mathcal{G}_1$ associated with a channel}\label{SecGraphG1}

We associate with each channel (given by its transition matrix) a graph which plays an important role in the proof of the metrization of the BAC channel.

\begin{definition}
Let $W:\mathcal{X} \rightarrow \mathcal{X}$ be a channel. The digraph $\mathcal{G}_1(W)$ has vertex set $\mathcal{X}^{(2)}$, the family of $2$-subsets of $\mathcal{X}$, and directed edges linking $\{i,j\}$ to $\{i,k\}$ when $\mbox{Pr}(i|j)<\mbox{Pr}(i|k)$.  
\end{definition}

\begin{example}\label{ExFirstExample}
Let $W:\mathcal{X} \rightarrow \mathcal{X}$ be a channel with $\mathcal{X}=\{a,b,c,d\}$ and transition matrix $$[W]=\left( \begin{array}{cccc} 0.4 & 0.3 & 0.1 & 0.2 \\ 0.1 & 0.2 & 0.1 & 0.1 \\ 0.2 & 0.1 & 0.3 & 0.1 \\ 0.3 & 0.4 & 0.5 & 0.6  \end{array}  \right).$$ Denoting by $xy$ the set $\{x,y\}$, the vertex set of $\mathcal{G}_{1}(W)$ is $\mathcal{X}^{(2)}=\{ab,ac,ad,bc,bd,cd\}$. To determine the edges we have to compare the conditional probabilities in each row of $[W]$ (without taking into account the main diagonal). The first row gives us the following information: $\mbox{Pr}(a|c)<\mbox{Pr}(a|d)<\mbox{Pr}(a|b)$ so, we obtain the following edges: $ac\rightarrow ad$, $ac\rightarrow ab$ and $ad \rightarrow ab$. Looking at the second row we have no inequalities among $\mbox{Pr}(b|a),\mbox{Pr}(b|c)$ and $\mbox{Pr}(b|d)$ so, there are no new edges among the vertices $ab,bc$ and $bd$. The third row gives us the inequalities: $\mbox{Pr}(c|b)<\mbox{Pr}(c|a)$ and $\mbox{Pr}(c|d)<\mbox{Pr}(c|a)$ which generate the following new edges $bc\rightarrow ac$ and $cd\rightarrow ac$. Finally, from the fourth row we obtain the inequalities: $\mbox{Pr}(d|a)<\mbox{Pr}(d|b)<\mbox{Pr}(d|c)$ from which we have the new edges $ad \rightarrow bd$, $ad\rightarrow cd$ and $bd \rightarrow cd$. Thus, the graph $\mathcal{G}_{1}(W)$ has $8$ edges and it is represented in Figure \ref{FigFirstExample}.
\end{example}

\begin{figure}[h]
\begin{center}
\includegraphics[scale=0.35]{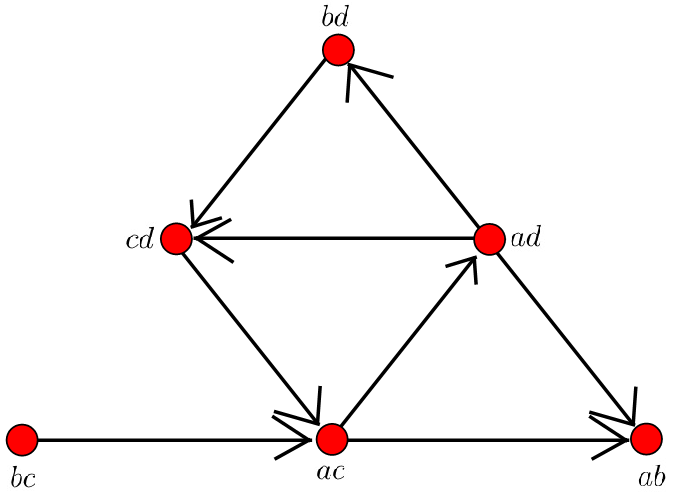} 
\end{center}
\caption{The digraph $\mathcal{G}_{1}(W)$ associated with the channel $W$ of Example \ref{ExFirstExample}.}
\label{FigFirstExample}
\end{figure}

A sufficient condition to guarantee the non-existence of a metric matched to a given channel $W$ is given in Proposition 5 of \cite{FW16}. This condition states that if the channel $W:\mathcal{X}\rightarrow \mathcal{X}$ admits a decision chain of length $r\geq 3$ it is not metrizable. A decision chain of length $r$ is a sequence $x_0,x_1,\ldots,x_{r-1} \in \mathcal{X}$ verifying $\mbox{Pr}(x_i|x_{i-1})<\mbox{Pr}(x_i|x_{i+1})$ for $i:0\leq i <r$, where the indices are taken modulo $r$ (we note that the definition given in \cite{FW16} in terms of $t$-decision region is equivalent to the one given here). Proposition 5 of the referred paper can be rewritten, in terms of the graph $\mathcal{G}_1(W)$, as follows.

\begin{proposition}\label{PropMetrizableAcyclic}
If a channel $W$ is metrizable, then its associated graph $\mathcal{G}_1(W)$ is acyclic. 
\end{proposition}

%%% PASADO HASTA ACA 12/04

The following example shows that the converse is false.

\begin{example}\label{ExAcyclicNonMetrizable}
Consider the channel $W:\mathcal{X}\rightarrow \mathcal{X}$ where $\mathcal{X}=\{0,1,2\}$ with transition matrix $$[W]=\left( \begin{array}{ccc} 1/2 & 1/4 & 1/4 \\ 15/36 & 25/36 & 5/9 \\ 1/12 & 1/18 & 7/36  \end{array}  \right).$$ In this case the graph $\mathcal{G}_{1}(W)$ has three vertices and only two edges: $\{0,1\}\rightarrow \{1,2\}$ and $\{1,2\}\rightarrow \{0,2\}$ therefore it is acyclic. However a metric compatible with $W$ should verify $d(0,1)=d(0,2)=d(2,0)<d(2,1)=d(1,2)<d(1,0)=d(0,1)$ which is impossible, therefore $W$ is not metrizable.
\end{example}

\begin{definition}
Let $W$ be a channel. The graph $\mathcal{G}_1(W)$ is transitive if for every path $(v_0, v_1 \ldots , v_{r-1})$ with $\# (v_0 \cap v_{r-1})=1$, the edge $v_0\rightarrow v_{r-1}$ belongs to $\mathcal{G}_{1}(W)$.  
\end{definition}

It is easy to see that if $\mathcal{G}_{1}(W)$ is transitive, then it is acyclic. The converse is false (the same channel of Example \ref{ExAcyclicNonMetrizable} provides a counterexample). Proposition \ref{PropMetrizableAcyclic} can be strengthened as follows.

\begin{proposition}\label{PropMetrizableTransitive}
If a channel $W$ is metrizable, then its associated graph $\mathcal{G}_1(W)$ is transitive. 
\end{proposition}

This proposition has straightforward verification (it can also be obtained as a particular case of Theorem \ref{MainTheorem1} in Section \ref{SecBAC}). Since every directed acyclic graph can be associated with a poset (in the way mentioned at the end of Section \ref{SecPreliminaries}), we can associate a poset with a channel whenever its associated graph is acyclic. The transitivity of the graph $\mathcal{G}_1(W)$ means that if $v<w$ and $v\cap w \neq \emptyset$, then $v \rightarrow w$ is and edge of $\mathcal{G}_1(W)$.\\

Let $W:\mathcal{X}\rightarrow \mathcal{X}$ be a channel. In terms of the conditional probabilities of $W$, the condition for the graph $\mathcal{G}_1(W)$ to be transitive can be written as follows: $\mathcal{G}_1(W)$ is transitive if and only if every sequence $x_0,x_1,\ldots,x_{r-1}\in \mathcal{X}$ ($r\geq 3$) satisfying $x_{i}\neq x_{i+1}$, $x_0\neq x_{r-1}$ and $\mbox{Pr}(x_{i}|x_{i-1})<\mbox{Pr}(x_{i}|x_{i+1})$ for $0\leq i \leq r-2$ (indices taken modulo $r$) also satisfies $\mbox{Pr}(x_{r-1}|x_{0})< \mbox{Pr}(x_{r-1}|x_{r-2})$. This is exactly the condition proposed in \cite{Poplawski16} to guarantee the existence of a metric $d$ such that 
\begin{equation}\label{EqUnaInclusion}
\mbox{arg }\max_{y\in C} \mbox{Pr}_{W}(x|y) \supseteq \mbox{arg }\min_{y\in C} d(x,y),
\end{equation} for all $C\subseteq \mathcal{X}$ and $x\in \mathcal{X}$ (interpreting both $\mbox{arg }\max$ and  $ \mbox{arg }\min$ as returning list of size at least $1$). Using this condition the author also proves that the BAC channels admit a metric verifying (\ref{EqUnaInclusion}). The reciprocal of Proposition \ref{PropMetrizableTransitive} is also false (in other words it is not possible to prove equality in equation (\ref{EqUnaInclusion}) under the hypothesis of transitivity).

\begin{example}\label{ExTransNonMetrizable}
Let $W:\mathcal{X}\rightarrow \mathcal{X}$ be the channel with alphabet $\mathcal{X}=\{0,1,2,3\}$ and matrix transition $$[W]= \left( \begin{array}{cccc} 0.44 & 0.22 & 0.22 & 0.11 \\ 0.26 & 0.52 & 0.26 & 0.13 \\ 0.12 & 0.08 & 0.16 & 0.04 \\ 0.18 & 0.18 & 0.36 & 0.72 \\
\end{array}  \right).$$ The graph $\mathcal{G}_{1}(W)$ is transitive. This graph and its Hasse diagram is showed in Figure \ref{TransitivoNM}. Every compatible metric should verify $d(2,1)>d(2,0)=d(0,2)=d(0,1)=d(1,0)=d(1,2)=d(2,1)$ which is impossible, then $W$ is not metrizable.
\end{example}

\begin{figure}[h]
\begin{center}
\includegraphics[scale=0.4]{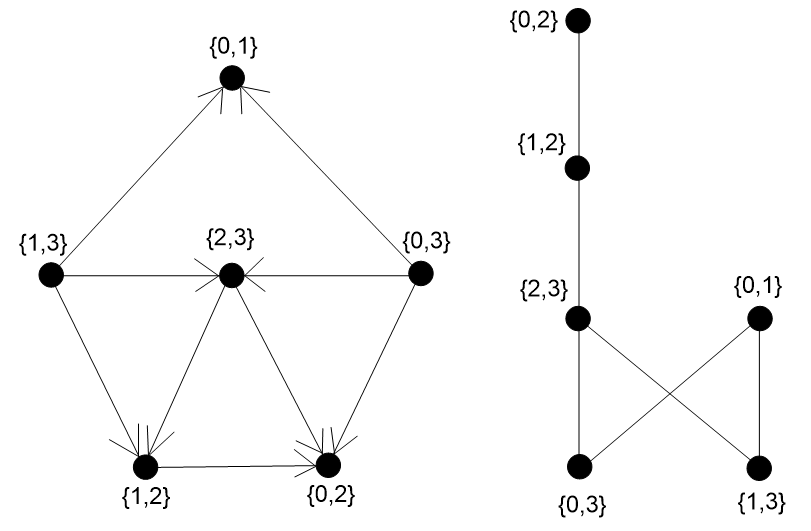} 
\end{center}
\caption{The transitive graph $\mathcal{G}_1(W)$ of Example \ref{ExTransNonMetrizable} (left) and its Hasse diagram (right).}
\label{TransitivoNM}
\end{figure}

\section{Colored posets and a necessary and sufficient condition for the channel to be metrizable}\label{SecColored}

Let $P$ be a poset. A chain of length $r\geq 0$ in $P$ is a finite sequence $(x_0,\ldots,x_r)$ such that $x_{i}<x_{i+1}$ for $0\leq i <r$. The height function relates to each element of a poset $P$, the maximum possible length of a chain ending in such element. In this paper we refer to this function as the standard height of $P$ and we call height function to any function $h:P \rightarrow \mathbb{N}$ verifying $h(x)<h(y)$ whenever $x<y$. A subset $X \subseteq P$ is called horizontal (with respect to $h$) when the restriction of $h$ to $X$ is constant. We say that a height is complete when its image is of the form $[k]=\{0,1,\ldots,k-1\}$ for some $k\in \mathbb{Z}^{+}$. To each height function $h$ we can associate a Hasse diagram such that $y$ is above $x$ if and only if $h(y)>h(x)$. This association establishes a bijection between complete heights and Hasse diagrams (in the sense that we can recover the height from its Hasse diagram).

As remarked in the previous section, when the graph $\mathcal{G}_1(W)$ associated with a channel $W:\mathcal{X}\rightarrow \mathcal{X}$ is acyclic, we can associate with it a poset $P=P(W)$ on the set $\mathcal{X}^{(2)}$. The standard height verifies $h(\{x,y\})<h(\{x,z\})$ whenever $\mbox{Pr}(x|y)<\mbox{Pr}(x|z)$ and from this, assuming the channel is reasonable (i.e. it verifies (\ref{EqReasonable})), we can construct a metric as in Proposition \ref{PropMetricFromh}, verifying $d(x,y)>d(x,z)$ whenever $\mbox{Pr}(x|y)<\mbox{Pr}(x|z)$. In particular this metric verifies (\ref{EqUnaInclusion}) and will be weakly metrizable in the sense of \cite{Poplawski16}. But this metric does not necessarily will match with the channel, since the poset structure of $W$ (when $\mathcal{G}_1(W)$ is transitive) does not give information about when two $2$-subsets $\{x,y\}$ and $\{x,z\}$ verify $\mbox{Pr}(x|y)=\mbox{Pr}(x|z)$ or not, except when they are connected in $\mathcal{G}_1(W)$. We need a more general structure to manage also with these cases.\\

Let $P$ be a poset. A coloration for $P$ is any function $c:P\rightarrow C$ ($C$ is a finite set) verifying that $c(x)\neq c(y)$ if $x<y$. A subset $X \subseteq P$ is monochromatic (with respect to $c$) if the restriction of $c$ to $X$ is constant. In particular every monochromatic set is an antichain of $P$. A colored cycle is a sequence $(x_0,x_1,\ldots,x_{r})$ in $P$ verifying that $x_0=x_r$ and $x_{i}<x_{i+1}$ if $c(x_{i})\neq c(x_{i+1})$ for $0\leq i <r$. Trivial examples of colored cycles are of the form $(x_0,x_1,\ldots,x_{r})$ with $x_0=x_r$ and $c(x_0)=c(x_1)=\cdots =c(x_r)$, we call these cycles monochromatic.

\begin{definition}
A colored poset is a pair $(P,c)$ where $P$ is a poset and $c$ a coloration for $P$ such that every colored cycle in $P$ is monochromatic.
\end{definition}

\begin{example} Consider $P=\{2,3,6,8,12,16\}$ with the divisibility relation (i.e. $x\leq y$ if $x$ divides $y$). Let $c_1$ and $c_2$ be the colorations for $P$ given by $c_1(2)=c_1(3)=\mbox{'blue'}, c_1(8)=c_1(6)=\mbox{'red'}, c_1(12)=c_1(16)=\mbox{'black'}$ and $c_2(2)=c_2(3)=\mbox{'blue'}, c_2(8)=c_2(12)=\mbox{'red'}, c_2(6)=c_2(16)=\mbox{'black'}$. Then, $(P,c_1)$ is a colored poset but $(P,c_2)$ is not, because it contains the non-monochromatic cycle $(8,16,6,12,8)$ (see Figure \ref{FigColoredPoset}).
\end{example}

\begin{figure}[h]
\begin{center}
\includegraphics[scale=0.4]{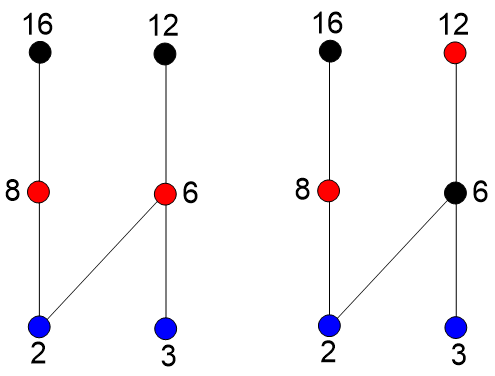} 
\end{center}
\caption{Two different colorations for the divisibility poset $P=\{2,3,6,8,12,16\}$; the first corresponds to a colored poset (left) and the second one does not (right).}
\label{FigColoredPoset}
\end{figure}

Let $P$ be a poset, $c:P\rightarrow C$ be a coloration for $P$ and $G_{P}$ be the digraph associated with $P$ (i.e. its vertex set is $P$ and $x\to y$ belongs to $G_{P}$ if $x<y$). We remark that a colored cycle is not necessarily a cycle in $G_{P}$, however colored cycles have a nice interpretation in term of graphs. Namely, if $G_{P}(c)$ denotes the digraph obtained from $G_{P}$ adding the edges $x\to y$ with $c(x)=c(y)$, the colored cycles of $P$ (with respect to $c$) correspond to cycles of $G_{P}(c)$.

A Hasse diagram for a colored poset $(P,c)$ is a Hasse diagram for $P$ with the additional property that if two points have the same color they are in the same level (i.e. no one is above or below the other). The next proposition guarantee the existence of a Hasse diagram for colored posets.

\begin{proposition}\label{PropExistenceHD}
Let $(P,c)$ be a colored poset. There exists a height function $h$ for $P$ such that every monochromatic subset of $P$ is horizontal with respect to $h$.
\end{proposition}

\begin{proof}
Let $c:P\rightarrow \{c_1,\ldots,c_k\}$ be the coloration, $A_i=c^{-1}(c_i)$ for $1\leq i \leq k$ and $P/c:=\{A_1,\ldots, A_k\}$. We note that $P/c$ is a partition of $P$ into monochromatic set, with $c(A_i)\neq c(A_j)$ if $i\neq j$ and every monochromatic subset of $P$ is contained in some $A_i$. Thus, it suffices to construct a height function $h$ for $P$ such that every $A_i$ is horizontal. For $X,Y \subseteq P$ we write $X< Y$ if there exist $x\in X$ and $y\in Y$ such that $x<y$. We claim that it is possible to order the indices of the elements of $P/c$ in such a way that if $A_i<A_j$ then $i<j$. Indeed, consider the digraph $G$ whose vertex set is $P/c$ and edges of the form $A_i\to A_j$ with $A_i<A_j$. This digraph is acyclic, because if there is a cycle $c=(A_{i_0},\ldots,A_{i_k})$ in $G$ with $k\geq 1$, then for each $j=1,\ldots,k-1$ there are $x_{i_j}\in A_{j}$ and $y_{i_{j+1}}\in A_{j+1}$ with $x_{i_j}<y_{i_{j+1}}$. Thus, the colored cycle $(x_0,y_1,x_1,y_2,\ldots,x_{k-1},y_{k},x_{0})$ is non-monochromatic (because $x_0<y_1$ implies $c(x_0)\neq c(y_1)$) which is a contradiction since $(P,c)$ is a colored poset. If $(P/c, \preccurlyeq)$ is the poset induced by the acyclic digraph $G$ (i.e. $A_i\preccurlyeq A_j$ if there is a path from $A_i$ to $A_j$ in $G$), clearly $A_i<A_j$ implies $A_i\preccurlyeq A_j$. By extending this poset to a total order we have $A_{i_1}\preccurlyeq A_{i_2}\preccurlyeq \cdots \preccurlyeq A_{i_k}$ where $i_1,i_2,\ldots,i_k$ is a permutation of $1,2,\ldots,k$. Thus we can assume that if $A_i<A_j$ then $i<j$ (ordering indices if necessary).

Next we define inductively an increasing sequence $h_1,h_2,\ldots,h_\kappa$ of heights for $P$ such that $A_j$ is horizontal with respect to $h_i$ if $j\leq i$. We write $x\geq A$ for $x \in P$ and $A \in P/c$ when $x\geq a$ for some $a\in A$ (otherwise we write $x\ngeq A$). We start considering the standard height $h_0$ and $t_1=\max\{h_{0}(x): x \in A_1\}$. We define $h_1$ as follows. 
$$h_1(x)=\left\{ \begin{array}{l}
\!\!\max\{h_0(x)+t_1-h_0(a): a \in A_1, a\leq x\}, \textrm{if }x\geq A_1, \\
\!\! h_0(x) \textrm{, otherwise.}
\end{array} \right.$$ We claim that this function is a height function for $P$ and $h_1(A_1)=\{t_1\}$ (in particular $A_1$ is horizontal with respect to $h_1$). Indeed, let $x,y\in P$ with $x>y$. We consider three cases: (i) $x>y\geq A_1$, (ii) $x\geq A_1$ and $y \ngeqslant A_1$ and (iii) $x,y \ngeqslant A_1$. In the first case, since for every $a\in A_1$ with $a\leq y$ we have $a\leq x$ and $h_0(y)+t_0-h_0(a)<h_0(x)+t_0-h_0(a)\leq h_1(x)$ for all $a\in A_1$ with $a\leq y$, then $h_1(y)< h_1(x)$. In the second and third cases we have $h_1(x)\geq h_0(x)>h_0(y)=h_1(y)$. In all the cases we conclude that $h_1(x)>h_1(y)$, thus $h_1$ is a height for $P$. Moreover, if $a\in A_1$ then $a\geq A_1$, and since $A_1$ is an antichain we have $h_1(a)=h_0(x)+t_1-h_0(a)=t_1$. 

Now we assume that there exists a height function $h_m$ for which $A_1,\ldots, A_m$ are horizontal ($1\leq m < \kappa$) and let $t_{m+1}=\max\{h_{m}(x): x \in A_{m+1}\}$. We define: $$h_{m+1}(x)= \max\{h_m(x)+t_{m+1}-h_m(a): a \in A_m \ , a\leq x\}$$ if $x\geq A_m$ and  $h_{m+1}(x)=h_m(x)$ otherwise. Using a similar argument to the case $m=1$ (considering three cases) we can prove that $h_{m+1}$ is a height function (i.e $h_{m+1}(x) > h_{m+1}(y)$ whenever $x>y$). Since $A_{m+1}$ is an antichain, then $h_{m+1}(a)=t_{m+1}$ for all $a \in A_{m+1}$. Let $x \in A_i$ for some $i$, $1\leq i \leq m$. We have $x \ngeqslant A_{m+1}$ because otherwise we would have $A_{i}\geq A_{m+1}$ with $i<m+1$ which is a contradiction. Therefore $h_{m+1}(a)=h_{m}(a)$ which, by inductive hypothesis, does not depend on $a\in A_i$. In the last step (when $m=k$) we obtain a height function $h_{k}$ for which all the elements of $P/c$ are horizontal. In particular, since every monochromatic subset is contained in some $A_i$, every monochromatic subset is horizontal with respect to $h_{k}$.
\end{proof}

\begin{remark}\label{RemarkHDisConstructive}
Since the proof of Proposition \ref{PropExistenceHD} is constructive, it brings us an algorithm to construct a Hasse diagram for a colored poset $(P,c)$. We start constructing the standard height function for $P$ (first step) and after at most $k$ steps we obtain a height function which induces a Hasse diagram for $(P,c)$, where $k$ is the number of colors. We say 'at most $k$ steps' instead of $k$ steps because when $A_{i+1}$ is horizontal with respect to $h_i$ we have $h_{i+1}=h_{i}$ (this happens for example when $A_{i+1}$ has a unique element) and we can omit this step.
\end{remark}

\begin{example}
Consider the colored poset $(P,c)$ where $P=\mathbb{Z}_{2}\times\mathbb{Z}_{3}$ with the order induced by $00<01<02, 10<11<12$ and $10<01$; and the coloration $c:P\rightarrow \{R,B,G,D\}$ given by $\widetilde{R}=\{00,11\}, \widetilde{B}=\{02,12\}, \widetilde{G}=\{01\}$ and $\widetilde{D}=\{10\}$, where $\widetilde{X}:=c^{-1}(X)$. The relation considered at the beginning of the proof of Proposition \ref{PropExistenceHD} restricted to $P/c$ is: $\widetilde{R}<\widetilde{G}, \widetilde{R}<\widetilde{B}, \widetilde{G}<\widetilde{B}$ and $\widetilde{D}<\widetilde{R}$, which can be extended to the total order $\widetilde{D}\preccurlyeq \widetilde{R}\preccurlyeq \widetilde{G}\preccurlyeq \widetilde{B}$. Therefore defining $A_1=\widetilde{D}, A_2=\widetilde{R}, A_3=\widetilde{G}$ and $A_4=\widetilde{B}$, we have that if $A_i<A_j$ then $i<j$. Figure \ref{FigConstructionHD} shows the different steps for the construction of a Hasse diagram for this colored poset. In the final stage we obtain the height function $h_4:P \rightarrow \mathbb{N}$ given by $h_4(10)=0, h_4(00)=h_4(11)=1, h_4(01)=2$ and $h_4(02)=h_4(12)=3$. 
\end{example}

\begin{figure}[h]
\begin{center}
\begin{tabular}{lll}
\includegraphics[scale=0.5]{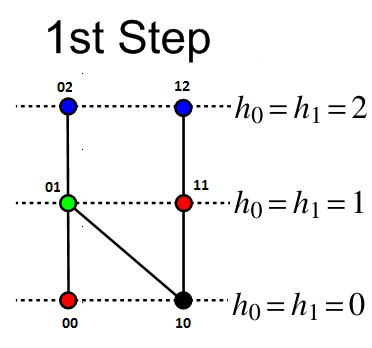} & 
\includegraphics[scale=0.4]{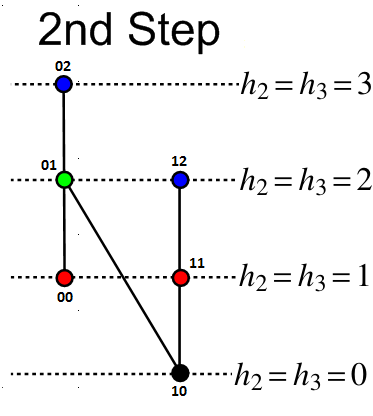} & \includegraphics[scale=0.4]{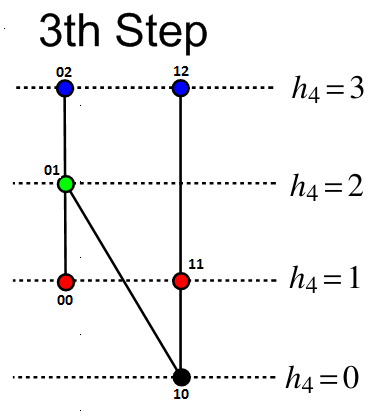}
\end{tabular}
\end{center}
\caption{The construction of a Hasse diagram for a colored poset.}
\label{FigConstructionHD}
\end{figure}

Our next goal is to associate with each channel $W$ (under certain conditions) a colored poset and to construct a metric from a height function for its Hasse diagram. We start by introducing the graphs $\mathcal{G}(W)$ and $\mathcal{G}_{0}(W)$ associated with the channel $W$.

\begin{definition}
Let $W:\mathcal{X}\rightarrow \mathcal{X}$ be a channel and $\mathcal{X}^{(2)}$ be the family of 2-subsets of $\mathcal{X}$. The digraph $\mathcal{G}(W)$ has vertex set $\mathcal{X}^{(2)}$ and directed edges linking $\{x,y\}$ to $\{x,z\}$ if $y\neq z$ and $\mbox{Pr}(x|y)\leq \mbox{Pr}(x|z)$. The graph $\mathcal{G}_{0}(W)$ is a non-directed graph whose vertex set is $\mathcal{X}^{(2)}$ and two vertices $\{x,y\}$ and $\{x,z\}$ are connected by an edge in $\mathcal{G}_{0}(W)$ if $y\neq z$ and $\mbox{Pr}(x|y)=\mbox{Pr}(x|z)$.
\end{definition}

The graph $\mathcal{G}_0(W)$ can be identified with the subgraph of $\mathcal{G}(W)$ whose vertex set is $\mathcal{X}^{(2)}$ and edges $\nu \to \omega$ and $\omega\to \nu$ for each edge $\{\nu,\omega\}$ in $\mathcal{G}_{0}(W)$. By construction the graphs $\mathcal{G}_0(W)$ and $\mathcal{G}_1(W)$ have no common edges. We use this identification in this paper.

\begin{definition}
A digraph $G$ is cycle-reverter if for each cycle $c=( v_0, v_1,   \ldots, v_{r-1}, v_{0})$ in $G$, then the reverse cycle $\overline{c}= ( v_0, v_{r-1},\ldots, v_{1}, v_{0})$ is also in $G$
\end{definition}

We remark that if $\mathcal{G}(W)$ is cycle-reverter then the graph $\mathcal{G}_1(W)$ is transitive (the converse is false) and, in particular, acyclic.

\begin{lemma}\label{LemmaColoredPosetChannel}
Let $W: \mathcal{X}\rightarrow \mathcal{X}$ be a channel such that its associated graph $\mathcal{G}(W)$ is cycle-reverter and let $\mathcal{A}=\{A_1,\ldots, A_{k}\}$ be the set of connected components of $\mathcal{G}_{0}(W)$. Consider for $P:=\mathcal{X}^{(2)}$ the poset structure induced by the graph $\mathcal{G}_1(W)$ and the function $c: P \rightarrow \mathcal{A}$ given by $c(v)=A$ if $v \in A$. Then $(P,c)$ is a colored poset.
\end{lemma}

\begin{proof}
First we prove that $c$ is a coloration for $P$. Let $v,w \in P$ such that $v<w$, then there exists a path $p_1=(v_{0}=v, v_1, \ldots , v_{r}= w)$ (with $r\geq 1$) in $\mathcal{G}_1(W)$. In particular $v=\{x,y\}$ and $v_1=\{x,z\}$ with $\mbox{Pr}(x|y)<\mbox{Pr}(x|z)$. We suppose, to the contrary, that $c(v)=c(w)$. Since $v$ and $w$ belong to the same connected component in $\mathcal{G}_0(W)$ there is a path $p_2$ from $w$ to $v$ in $\mathcal{G}_{0}(W)$. Since $\mathcal{G}(W)$ is cycle-reverter, the reverse of the cycle $p_1*p_2$ is also in $\mathcal{G}(W)$. In particular the arrow $v_1 \rightarrow v$ is in $\mathcal{G}(W)$ and then $\mbox{Pr}(x|z)\geq \mbox{Pr}(y|z)$, which is a contradiction. Now we prove that $(P,c)$ is a colored poset. Consider a colored cycle $C=(v_0,v_1,\ldots,v_r)$ with $v_0=v_r$.
If $v_i<v_{i+1}$ there is a path from $v_i$ to $v_{i+1}$ in $\mathcal{G}_1(W)$ and if $c(v_i)=c(v_{i+1})$ there is a path from $v_i$ to $v_{i+1}$ in $\mathcal{G}_0(W)$. In both cases there is a 
path $p_i$ from $v_i$ to $v_{i+1}$ in $\mathcal{G}(W)$ for $0\leq i <r$. Since the graph $\mathcal{G}(W)$ is cycle-reverter the reverse of the cycle $p_1*p_2* \ldots *p_{r-1}$ is a cycle in $\mathcal{G}(W)$. Therefore none of the paths $p_i$ can be in $\mathcal{G}_1(W)$ and we conclude that all the vertices $v_i$ belong to the same connected component in $\mathcal{G}_0(W)$, then $C$ is monochromatic. 
\end{proof}

The next theorem establishes a necessary and sufficient condition for the existence of a metric matched to a given channel.

\begin{theorem}\label{MainTheorem1}
Let $W: \mathcal{X}\rightarrow \mathcal{X}$ be a channel. The graph $\mathcal{G}(W)$ is cycle-reverter if and only if the channel $W$ is metrizable.
\end{theorem}

\begin{proof}
First we suppose the existence of a metric $d:\mathcal{X}\times \mathcal{X}\rightarrow [0,+\infty)$ matched to $W$ and consider a cycle $c=(v_0, v_1, \ldots, v_r= v_{0})$ in $\mathcal{G}(W)$. Every vertex is of the form $v_{i}=\{x_i,y_i\}$ with $x_i,y_i \in \mathcal{X}$ and $\# (v_{i}\cap v_{i+1})=1$ for $0\leq i < r$ (indices taken modulo $r$). Since $d$ matches with $W$, from the cycle $c$, we obtain the following chain of inequalities: $$d(x_0,y_0)\leq d(x_1,y_1)\leq \cdots \leq d(x_{r-1},y_{r-1}) \leq d(x_0,y_0).$$ Therefore every inequality is actually an equality and the reverse cycle $\overline{c}$ is also in $\mathcal{G}(W)$. This proves that the graph $\mathcal{G}(W)$ is cycle-reverter whenever $W$ is metrizable. Conversely, if $\mathcal{G}(W)$ is cycle-reverter then by Lemma \ref{LemmaColoredPosetChannel} we can define in $P=\mathcal{X}^{(2)}$ a colored poset structure where the order is induced by the graph $\mathcal{G}_1(W)$ and the connected components of $\mathcal{G}_0(W)$ are monochromatic. By 
Proposition \ref{PropExistenceHD} we can construct a height function $h$ for $P$ such that every connected component of $\mathcal{G}_0(W)$ is horizontal. In particular $\mbox{Pr}(x|y)<\mbox{Pr}(x|z)$ if and only if $h(\{x,y\})<h(\{x,z\})$. Hence, the function $(x,y)\mapsto h(\{x,y\})$ is coherent-with-$W$, then $W$ is metrizable (a metric can be constructed as in Equation (\ref{EqMetricFromf})).
\end{proof}

\begin{remark}
When $\mathcal{G}(W)$ is cycle-reverter we have the following algorithm to obtain a metric $d$ matching to the channel $W: \mathcal{X}\rightarrow\mathcal{X}$.
\begin{enumerate}
\item Consider the colored poset in $P=\mathcal{X}^{(2)}$ whose partial order is induced by $\mathcal{G}_1(W)$ and the coloring is given by the connected components of $\mathcal{G}_0(W)$.
\item Construct a height function $h$ for $P$ as in the proof of Proposition \ref{PropExistenceHD} (see also Remark \ref{RemarkHDisConstructive}) for which every monochromatic set is horizontal.
\item Let $m$ be the maximum value of $h$. By Proposition \ref{PropMetricFromh}, a metric matched to $W$ is given by $$d(x,y)=\left\{\begin{array}{ll}
2m-h\left(\{x,y\}\right) & \textrm{if }x\neq y \\
0 & \textrm{if }x=y
\end{array}  \right.$$
\end{enumerate}
\end{remark}

\begin{example}\label{ExGettingaMetric}
Let $\mathcal{X}=\{0,1,2,3\}$ and $W:\mathcal{X} \rightarrow \mathcal{X}$ be the channel with transition matrix: $$[W]=\left( \begin{array}{cccc} 4/9 & 0 & 0 & 2/9 \\ 0 & 2/9 & 1/9  & 1/9 \\ 2/9  & 4/9 & 8/9 & 0 \\ 1/3 & 1/3 & 0  & 2/3  \end{array} \right).$$ We consider for $\mathcal{X}^{(2)}$ the order induced by $\mathcal{G}_{1}(W)$ and the coloration $c$ as in Lemma \ref{LemmaColoredPosetChannel} (i.e. each color correspond to a connected component of $\mathcal{G}_{0}(W)$). Figure \ref{FigMetrizationOfW} shows the graph $\mathcal{G}(W)$ where the edges corresponding to the subgraphs $\mathcal{G}_{1}(W)$ and $\mathcal{G}_{0}(W)$ are colored black and red respectively. Colored cycles correspond to cycles in $\mathcal{G}(W)$. Note that  a cycle and its reverse are in $\mathcal{G}(W)$ if and only if it is a cycle in $\mathcal{G}_{0}(W)$. Since there are no cycles in $\mathcal{G}(W)$ containing black edges, the graph $\mathcal{G}(W)$ is cycle-reverter. Thus, by Lemma \ref{LemmaColoredPosetChannel}, $(\mathcal{X}^{(2)},c)$ is a colored poset. We can apply the steps given in the proof of Proposition \ref{PropExistenceHD} to obtain a Hasse diagram for this colored poset. This process is illustrated in Figure \ref{FigMetrizationOfW}, after three steps we obtain the height function $h_3: \mathcal{X}^{(2)}\rightarrow \mathbb{N}$ given by $h_3\left( \{2,3\} \right)=0$, $h_3\left( \{0,1\} \right)=h_3\left( \{0,2\} \right)=1$ and $h_3\left( \{1,3\} \right)=h_3\left( \{0,3\} \right)=h_3\left( \{1,2\} \right)=3$. A metric matched to $W$ is given by $d(x,y)=6-h_3(\{x,y\})$ when $x\neq y$ and $0$ otherwise.
\end{example}
\vspace{-4mm}

\begin{figure}[h]
\begin{center}
\begin{tabular}{ll}
\includegraphics[scale=0.4]{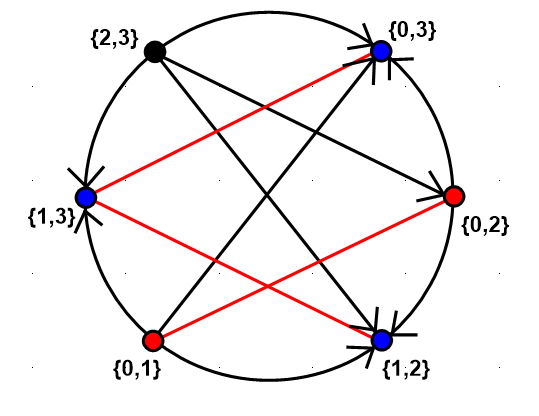} & \includegraphics[scale=0.4]{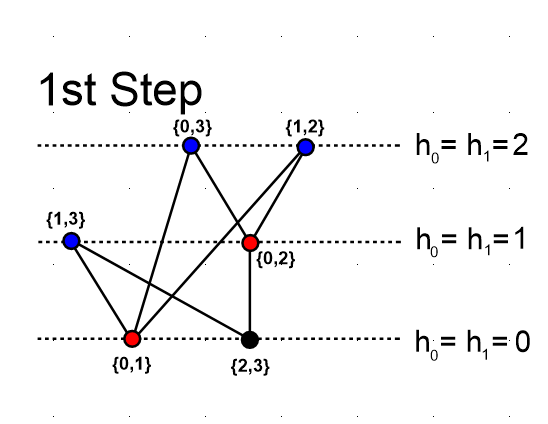} \\
\includegraphics[scale=0.4]{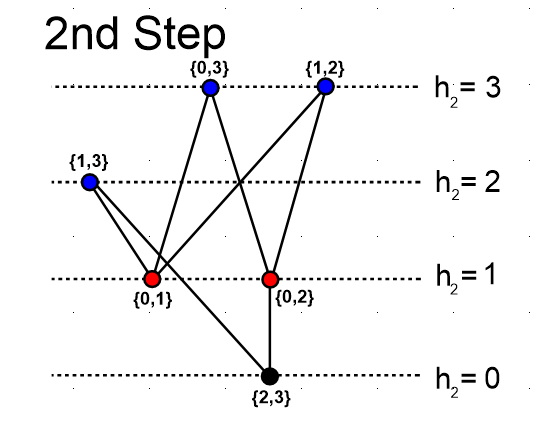} & 
\includegraphics[scale=0.4]{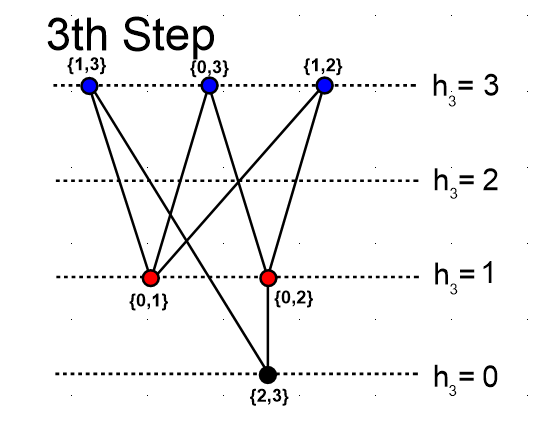}
\end{tabular}
\end{center}
\caption{The graph $\mathcal{G}(W)$ for the channel of Example \ref{ExGettingaMetric} and the steps to obtain a metric matched to this channel.}
\label{FigMetrizationOfW}
\end{figure}

\section{The BAC channel is metrizable}\label{SecBAC}

We consider the $n$-fold BAC channel $BAC^n(p,q)$ with parameters $p, q \in [0,1]$ and $p+q<1$. The case $pq=0$ corresponds to the $Z$-channels which we know they are metrizable (Theorem 6 of \cite{FW16}), then we can assume $pq>0$. Each entry of its transition matrix $M_n(p,q)$ is of the form $$\mbox{Pr}(x|y)=p^a(1-p)^b q^c(1-q)^d$$ with $a+b+c+d=n$ and $a+d=w(x)$ (the Hamming weight of $x\in \mathbb{Z}_2^{n}$). If we consider other word $y'\in\mathbb{Z}_2^{n}$ with $\mbox{Pr}(x|y')=p^{a'}(1-p)^{b'} q^{c'}(1-q)^{d'}$, since $a+d=a'+d'$ and $b+c=b'+c'$, taking the quotient we have:
\begin{equation}\label{EqCociente}
\frac{\mbox{Pr}(x|y)}{\mbox{Pr}(x|y')}= \left(\frac{1-p}{q} \right)^{b-b'}\cdot \left(\frac{1-q}{p} \right)^{d-d'}.
\end{equation}
This identity will be useful in our proof of metrizability of the BAC channel.

\begin{lemma}\label{LemmaMinimal}
Let $W:\mathcal{X}\rightarrow \mathcal{X}$ be a channel. The graph $\mathcal{G}(W)$ is cycle-reverter if and only if every sequence $x_0,x_1,\ldots,x_{r-1}\in \mathcal{X}$ ($r\geq 3$) satisfying $x_i\neq x_{i+1}$ and $\mbox{Pr}(x_{i}|x_{i-1})\leq \mbox{Pr}(x_{i}|x_{i+1})$ for $i:0\leq i <r$ also satisfy $\mbox{Pr}(x_{i}|x_{i-1})= \mbox{Pr}(x_{i}|x_{i+1})$ for $i:0\leq i <r$ (where the indices are considered modulo $r$).
\end{lemma}

\begin{proof}
We suppose that $\mathcal{G}(W)$ is cycle-reverter and consider a sequence $x_0,x_1,\ldots,x_{r-1} \in \mathcal{X}$ satisfying $x_i\neq x_{i+1}$ and $\mbox{Pr}(x_{i}|x_{i-1})\leq \mbox{Pr}(x_{i}|x_{i+1})$ for $i:0\leq i <r$. Then we have a cycle $c=(v_0, \ldots , v_{r-1}, v_{r} =v_{0})$ in $\mathcal{G}(W)$ given by $v_i=\{x_i,x_{i+1}\}$ for $0\leq i <r$. Since this graph is cycle-reverter its reverse cycle $\overline{c}$ is also a cycle in $\mathcal{G}(W)$ which implies $\mbox{Pr}(x_{i}|x_{i-1})= \mbox{Pr}(x_{i}|x_{i+1})$ for $i:0\leq i <r$. Now we suppose that the graph $\mathcal{G}(W)$ is not cycle-reverter. In this case we can find a sequence $x_0,x_1,\ldots,x_{r-1}\in \mathcal{X}$ satisfying $x_i\neq x_{i+1}$ and $\mbox{Pr}(x_{i}|x_{i-1})\leq \mbox{Pr}(x_{i}|x_{i+1})$ for $i:0\leq i <r$ where at least one inequality is strict. Indeed, consider a cycle $c=( v_0, v_1, \ldots, v_r= v_0)$ in $\mathcal{G}(W)$ of minimal length $r\geq 1$ whose reverse cycle $\overline{c}$ is not in this graph. Since $c=\overline{c}$ for cycles of length $r\leq 2$, we have $r\geq 3$. We remark that the fact that its reverse cycle $\overline{c}$ is not in $\mathcal{G}(W)$ is equivalent to the existence of some arrow in $c$ which is also an arrow in $\mathcal{G}_1(W)$. The vertices in $c$ are pairwise disjoint since otherwise we could take a sub-cycle of $c$ containing some arrow of $\mathcal{G}_1(W)$ contradicting the minimality of $r$. If for some $i:0\leq i <r$ we have that $v_{i}\cap v_{i+1}\cap v_{i+2} = \{x\}$, then there exists $y,z,t$ pairwise distinct such that $v_{i}=\{x,y\}$, $v_{i+1}=\{x,z\}$ and $v_{i+2}=\{x,t\}$. Thus $\mbox{Pr}(x|y)\leq \mbox{Pr}(x|z)\leq \mbox{Pr}(x|t)$ and the arrow $v_{i}\rightarrow v_{i+2}$ is also in $\mathcal{G}(W)$. If some of the arrows $v_{i}\rightarrow v_{i+1}$ or $v_{i+1}\rightarrow v_{i+2}$ is in $\mathcal{G}_1(W)$ then $v_{i}\rightarrow v_{i+2}$ is also in $\mathcal{G}_1(W)$, so we could substituting these two arrows for the last obtaining a new cycle, whose reverse is not in $\mathcal{G}(W)$ and length $r-1$ which contradict the minimality of $r$. Therefore $v_{i}\cap v_{i+1}\cap v_{i+2} = \emptyset$ for all $i:0\leq i <r$ and there exists a sequence $x_0,x_1,\ldots,x_{r-1}\in \mathcal{X}$ such that $v_i=\{x_i,x_{i+1}\}$. This sequence satisfies $\mbox{Pr}(x_{i}|x_{i-1})\leq \mbox{Pr}(x_{i}|x_{i+1})$ for $i:0\leq i <r$ with at least one strict inequality (the corresponding to the edge in $\mathcal{G}_1(W)$).
\end{proof}

\begin{theorem}\label{MainTheorem2}
Let $n\geq 2$ and $(p,q)\in (0,1]^2$ with $p+q<1$. Then, the channel $W=BAC^n(p,q)$ is metrizable. 
\end{theorem}

\begin{proof}
By Theorem \ref{MainTheorem1}, it is enough to prove that its associated graph $\mathcal{G}(W)$ is cycle-reverter. We assume, to the contrary, that this graph is not cycle-reverter and by Lemma \ref{LemmaMinimal} there exists a sequence $x_0,x_1,\ldots,x_{r-1} \in \mathcal{X}$ such that $x_i\neq x_{i+1}$ and
\begin{equation}\label{EqLessOrEqual}
\mbox{Pr}(x_i|x_{i-1})\leq \mbox{Pr}(x_i|x_{i+1}), \quad \forall i: 0\leq i < r,
\end{equation}
where the indices are taken modulo $r$ and where at least one of these inequality is strict. We write these conditional probability as $\mbox{Pr}(x_i|x_{i-1})=p^{a_i}(1-p)^{b_i}q^{c_i}(1-q)^{d_i}$ for $i:0\leq i <r$. Therefore $\mbox{Pr}(x_{i-1}|x_{i})=p^{c_i}(1-p)^{b_i}q^{a_i}(1-q)^{d_i}$ for $i:0\leq i <r$ and applying Equation \ref{EqCociente} we obtain
$$\frac{\mbox{Pr}(x_i|x_{i+1})}{\mbox{Pr}(x_i|x_{i-1})}= \frac{p^{c_{i+1}}(1-p)^{b_{i+1}}q^{a_{i+1}}(1-q)^{d_{i+1}}}{p^{a_i}(1-p)^{b_i}q^{c_i}(1-q)^{d_i}}$$
\begin{equation}\label{EqCociente2}
\qquad \qquad \quad =\left(\frac{1-p}{q} \right)^{b_{i+1}-b_{i}} \left(\frac{1-q}{p} \right)^{d_{i+1}-d_{i}},
\end{equation}
for $0\leq i <r$. Multiplying these $r$ inequalities we have 
$$\prod_{i=0}^{r-1} \frac{\mbox{Pr}(x_i|x_{i+1})}{\mbox{Pr}(x_i|x_{i-1})} =  \prod_{i=0}^{r-1} \left(\frac{1-p}{q} \right)^{b_{i+1}-b_{i}} \left(\frac{1-q}{p} \right)^{d_{i+1}-d_{i}} $$ $$ =  \left(\frac{1-p}{q} \right)^{\Sigma_{i=0}^{r-1}(b_{i+1}-b_{i})} \left(\frac{1-q}{p} \right)^{\Sigma_{i=0}^{r-1}(d_{i+1}-d_{i})}=1$$ because $b_{r}=b_{0}, d_{r}=d_{0}$ and $\Sigma_{i=0}^{r-1}(b_{i+1}-b_{i})=\Sigma_{i=0}^{r-1}(d_{i+1}-d_{i})=0$. But by (\ref{EqLessOrEqual}) this product is greater than $1$ (since at least one inequality is strict) which is a contradiction. Therefore $W$ is metrizable.
\end{proof}

\section{Concluding remarks and further problems}\label{SecFinal}

In this work we approach the problem of metrization for the $n$-fold BAC channels in the sense of the definition used in \cite{Seguin80} and \cite{FW16}.  
An existence proof and an algorithm to construct a metric matching to the BAC channels are provided. An interesting problem is to describe the set $D_1(n,p,q)$ of metrics matching to the channel $BAC^{n}(p,q)$. This set is non-empty by Theorem \ref{MainTheorem2} and closed under linear combinations with positive real coefficients, so we could look for a minimal generator for this set. Describing this set allows to choose good metrics according to a given criterion. One possible criterion could be to select the metric according to how easy is to compute it. Other possible criterion could be to select the metric according to how good it fits the channel in the sense of the next definition.
\begin{definition}
Let $W:\mathcal{X}\rightarrow\mathcal{X}$ be a channel and $d$ be a metric compatible to $W$. The metric $d$ is matched to the channel $W$ with order $n$ if $d^m$ is a metric matched to $W^m$ for all $m: 1\leq m \leq n$, where $d^m: \mathcal{X}^m \rightarrow \mathcal{X}^m$ is given by 
$$d^m(x,y)=\sum_{k=1}^m d(x_i,y_i)\qquad \textrm{and} \qquad \mbox{Pr}_{W^{m}}(x|y)=\prod_{k=1}^m \mbox{Pr}_{W}(x_i|y_i).$$  If $d$ is matched to $W$ with order $n$ for all $n\geq 1$, we say that the metric $d$ matches completely to $W$.
\end{definition}

We also define the order of metrizability of $W$ as the maximum $n$ (possibly infinite) for which there exists a metric matched to $W$ with order $n$. It would be interesting to determine the order of metrizability for the BAC channels or at least to determine which of these channels admit a matched metric with order $n\geq 2$. In \cite{Seguin80} it was approached the problem of determining when a channel $W:\mathcal{X}\rightarrow \mathcal{X}$ is completely metrizable for alphabets of lengths $\#\mathcal{X}=2,3$ and was proved that the channel $BAC^{1}(p,q)$ has a matched metric with order $\infty$ if and only if $p=q$ (symmetric channel). In this case the Hamming metric matches completely to the channel. 

\section*{acknowledgements}
The author would like to thank the anonymous reviewers for their valuable comments and suggestions, Sueli Costa for her support and suggestions, and CNPq and FAPESP for their support.

\end{document}